\documentclass[12pt]{article}
\usepackage{amsmath}
\usepackage{graphicx,psfrag,epsf}
\usepackage{enumerate}
\usepackage{natbib}
\usepackage{url} 

\usepackage{algorithm}
\usepackage{algorithmic}
\usepackage{amssymb}
\usepackage{bm}

\newtheorem{Thm}{Theorem}
\newtheorem{lemma}{Lemma}
\newtheorem{proof}{Proof}

\newcommand{\blind}{1}

\begin{document}

\def\spacingset#1{\renewcommand{\baselinestretch}%
{#1}\small\normalsize} \spacingset{1}


\if1\blind
{
  \title{\bf Combining Forecasts Using Ensemble Learning}
  \author{Hamed Masnadi-Shirazi \hspace{.2cm}\\
    School of Electrical and Computer Engineering \\Shiraz University \\Shiraz, Iran}
  \maketitle
} \fi

\if0\blind
{
  \bigskip
  \bigskip
  \bigskip
  \begin{center}
    {\LARGE\bf Combining Forecasts Using Ensemble Learning}
\end{center}
  \medskip
} \fi

\bigskip
\begin{abstract}
The problem of combining individual forecasters to produce a forecaster with improved performance is considered. The connections between probability elicitation and classification are used to pose the combining forecaster problem as that of ensemble learning. With this connection in place, a number of theoretically sound ensemble learning methods such as Bagging and Boosting are adapted for combining forecasters. It is shown that the simple yet effective method of averaging the forecasts is equivalent to Bagging. This provides theoretical insight into why the well established averaging of forecasts method works so well. Also, a nonlinear combination of forecasters can be attained through Boosting which is shown to theoretically produce combined forecasters that are both calibrated and highly refined. Finally, the proposed methods of combining forecasters are applied to the Good Judgment Project data set and are shown to outperform the individual forecasters. 
\end{abstract}

\noindent%
\vfill

\newpage
\spacingset{1.45} 

\section{Introduction}
\label{intro}
probability forecasting and probability elicitation \citep{Savage,DeGroot} are important in many fields such as weather forecasting where a forecaster provides the probability of rain on a certain day \citep{Tilmann2007,Brocker2009}
or in medicine where a physician might provide the probability of a patient surviving a certain amount of time \citep{Winkler}
or in economics where the probability of a certain economic variable is provided \citep{Croushore}
or in politics where an analyst might provide the probability of a certain political event such as wining the election \citep{GoodJudgmentProj}.

Improving a forecaster's performance is an obvious goal and there is strong empirical evidence that combining different forecasts  of the same phenomenon
can greatly improve performance \citep{Timmermann1, Clements}. Intuitively, this can be understood as combining the information available to all forecasters or averaging the results to form an improved forecast. There are two interconnected questions that remain. How do we exactly combine the different forecasts and why would combing forecasters theoretically guarantee improved performance given the method. 

Most methods of combining forecasters entail  a weighted linear combination of the individual forecasters. The simplest is the averaging method which reports the average of the multiple forecasts. While many successful applications of the the simple averaging method exists in  weather forecasting \citep{Baars}, medical forecasting
\citep{Winkler} and economics \citep{Graham}, there has been recent theoretical work that suggests that a linear combination of forecasters is suboptimal \citep{Ranjan}.
Nevertheless, many different methods have been proposed for improving the average method \citep{Satopaa1,Clements}.  

In this paper we take a different  novel approach . We first recount the connections between probability elicitation and  classification 
\citep{Buja,HamedNunoJMLRRegularize}. With this connection in place, we borrow the ensemble learning techniques for combining classifiers and use them for producing novel methods of combining forecasters. The goal in ensemble learning is to combine multiple classifiers to form an ensemble classifier with improved classification performance. Different methods such as Bagging \citep{Bagging}
and Boosting \citep{freund, friedman} have been proposed with strong theoretical guarantees. 

We adapt the Bagging and Boosting algorithms to the problem of  combining forecasters with a number of interesting consequences. 
First, we show that the adapted Bagging algorithm is equivalent to the well established method of averaging the forecasters. 
The interesting aspect of this result is that Bagging has a strong theoretical justification that can be readily applied to provide 
much needed insight into why the simple method of averaging the forecasts works so well in many cases.

Second, we show that the adapted Boosting algorithm leads to a novel nonlinear method for combining forecasters. 
We borrow the strong theoretical results that back the Boosting algorithm and establish a connection between the notions of evaluating a classifier through risk
minimization \citep{friedman, HamedNunoLossDesign} and evaluating a forecaster by calibration and refinement (sharpness) \citep{Savage,DeGroot}. Specifically, we show that a combined forecaster found using the  Boosting algorithm is both calibrated and highly refined. 

Finally, we build a combined forecaster for the Good Judgment Project \citep{GoodJudgmentProj} using the adapted Bagging and Boosting algorithms. The Good Judgment Project data set involves predicting the outcome  of a  variety of mostly economic and political questions.
We show that a combined forecaster using Boosting has improved performance and can correctly predict the outcome of $338$ questions with only $6$ errors.

\section{The Connection Between Probability Elicitation and Classification }
Given that ensemble methods are commonly used for classification, we first review the general connections between classification and probability elicitation.

\subsection{The Classification Problem and Risk Minimization  }
\label{sec:DeriveTangentLoss}
Building  a classifier by way of minimizing a risk functional is a well established method ~\citep{friedman,zhang,Buja,HamedNunoLossDesign}.
We define a \emph{classifier} $h$  as a mapping from a feature vector ${\bf x} \in \cal X$ to a class 
label $y \in \{-1,1\}$. It is assumed that the feature vectors ${\bf x}$ are sampled from the probability distribution $P_{\bf X}({\bf x})$ and the class labels $y$  are sampled from the probability distribution  $P_{Y|X}(y|{\bf x})$. A classifier is constructed  by taking the sign of the classifier \emph{predictor} $p: {\cal X} \rightarrow \mathbb{R}$ and is written as 
\begin{equation}
  h({\bf x}) = sign[p({\bf x})].
  \label{eq:h}
\end{equation}

We define the \emph{risk} as 
\begin{equation}
  R(p) = E_{{\bf X},Y}[L(p({\bf x}),y)]
  \label{eq:risk}
\end{equation}
for a  non-negative loss function $L(p({\bf x}),y)$.
The optimal predictor $p^*({\bf x})$ is found by minimizing the risk or equivalently the conditional risk 
\begin{equation*}
  E_{Y|{\bf X}} [L(p({\bf x}),y)|{\bf X} = {\bf x}]
\end{equation*}
for all ${\bf x} \in {\cal X}$.

The predictor $p({\bf x})$ can itself be partitioned as
\begin{equation*}
  p({\bf x}) = f(\eta({\bf x}))
  \label{eq:compose}
\end{equation*}
into the  \emph{link function} $f: [0,1] \rightarrow \mathbb{R}$   and  the posterior probability function $\eta({\bf x}) = P_{Y|{\bf X}}(1|{\bf x})$.
This allows us to learn the optimal predictor  by first analytically finding the  link $f(\eta)$
and then estimating $\eta({\bf x})$. We say that a link function $f^*(\eta)$ is optimal if it is   a one-to-one function that implements the Bayes decision rule as 
\begin{equation}
  \left\{
  \begin{array}{cc}
    f^*(\eta) > 0 & \mbox{if $\eta > \frac{1}{2} $} \\
    f^*(\eta) = 0 & \mbox{if $\eta =  \frac{1}{2} $} \\
    f^*(\eta) < 0 & \mbox{if $\eta <  \frac{1}{2} $.}
  \end{array}
  \right.
  \label{eq:Bayesnec}
\end{equation}
It can be shown that   $f^*$ 
minimizes the    conditional zero-one risk of
\begin{eqnarray}
  C_{0/1}(\eta,p) = \eta \frac{1- sign(p)}{2} + 
  (1-\eta) \frac{1 + sign(p)}{2}
    = \left\{ \begin{array}{ll}
         1-\eta, & \mbox{if $p=f(\eta) \geq 0 $};\\
        \eta, & \mbox{if $p=f(\eta)<0$} \end{array} \right.
\end{eqnarray}
 associated with the zero-one loss of
\begin{eqnarray*}
  L_{0/1}(y,p) = \frac{1- sign(yp)}{2}  = \left\{ \begin{array}{ll}
         0, & \mbox{if $y=sign(p)$};\\
        1, & \mbox{if $y \ne sign(p)$}.\end{array} \right.
\end{eqnarray*}
Two examples of such optimal link functions are
\begin{eqnarray}
\label{eq:fexamples}
&&f^*=\log\frac{\eta}{1-\eta} \\
&&f^*=2\eta-1. \nonumber  \\
\end{eqnarray}

The resulting classifier written as
\begin{equation}
h^*({\bf x}) = sign[f^*(\eta({\bf x}))],
\label{eq:ResultClassifier}
\end{equation}
is equivalent to the optimal Bayes decision rule with associated   minimum zero-one conditional  risk 
\begin{eqnarray}
\label{eq:zeronemincondrisk}
  C^*_{0/1} (\eta) &&= \eta\left(\frac{1}{2}-\frac{1}{2}sign(2\eta-1)\right)+ 
           (1-\eta)\left(\frac{1}{2}+\frac{1}{2}sign(2\eta-1)\right) \\
					&&=\left\{ \begin{array}{ll}
         (1-\eta) & \mbox{if $\eta \geq \frac{1}{2} $}\\
        \eta & \mbox{if $\eta<\frac{1}{2}$}\end{array} \right. \\
           &&=\min\{\eta, 1-\eta\}
\end{eqnarray}
and minimum zero-one risk denoted by $R_{0/1}(p^*)$ which is  also known as the Bayes error ~\citep{JordanBartlett, zhang, book:ProbPatRec}.

Loss functions other than the zero-one loss can be used. \emph{Margin losses} are commonly used in classification algorithms and are in the form of 
\begin{eqnarray}
L_{\phi}(f,y)=\phi(yf).
\end{eqnarray}
Unlike the zero-one loss, they are usually differentiable and assign nonzero loss to correctly classified points.  
The conditional risk of a margin loss is
\begin{equation}
  C_\phi(\eta,p) = C_\phi(\eta,f(\eta)) = \eta \phi(f(\eta)) + 
  (1-\eta) \phi(-f(\eta)),
  \label{eq:CondRisk}
\end{equation}
and is minimized by the optimal link
\begin{equation}
  f^*_{\phi}(\eta) = \arg\min_{f} C_\phi(\eta,f).
  \label{eq:fstarphi}
\end{equation}
The resulting minimum conditional risk function is
\begin{equation}
  C^*_\phi(\eta) = C_\phi(\eta,f^*_\phi).
  \label{eq:C*phi}
\end{equation}

The optimal link  can be found analytically for most margin losses.
In such cases, the optimal predictor of minimum risk is $p^* = f^*_\phi(\eta)$ and the  minimum risk is 
\begin{equation}
\label{equ:MinRiskInit}
R_{\phi}({p^*}) = \int_{\bf x} P_{{\bf X}}({\bf x}) \left[ P_{{\bf Y}|{\bf X}}(1|{\bf x})\phi({ p^*}({\bf x})) + P_{{\bf Y}|{\bf X}}(-1|{\bf x})\phi(-{ p^*}({\bf x}))  \right] d{\bf x}.
\end{equation}
If $f^*_\phi(\eta)$ is  invertible then we have a \emph{proper loss} and the posterior probabilities $\eta$ can be found using 
\begin{equation}
  \eta({\bf x}) = [f^*_\phi]^{-1}( p^*({\bf x})).
  \label{eq:link}
\end{equation} 
Additionally,  $C_\phi^*(\eta)$ is concave \citep{zhang} and 
\begin{eqnarray}
  C_\phi^*(\eta) &=&   C_\phi^*(1-\eta) \label{eq:Cstarsym}.
\label{eq:fstarsym}
\end{eqnarray}

For example the \emph{exponential loss} is 
\begin{equation}
\phi(yp({\bf x}))=e^{-yp({\bf x})}
\end{equation}
with optimal link
\begin{equation}
f^*_{\phi}(\eta)=\frac{1}{2}\log(\frac{\eta}{1-\eta}),
\end{equation}
invertable inverse link
\begin{equation}
(f^*)^{-1}(v)=\frac{e^{2v}}{1+e^{2v}},
\end{equation}
and concave minimum conditional risk 
\begin{equation}
C^*_{\phi}(\eta)=2\sqrt{\eta(1-\eta)}.
\end{equation}

In practice, the empirical risk 
\begin{equation}
  R_{emp}(p) = \frac{1}{n} \sum_i L(p({\bf x}_i),y_i)
  \label{eq:emprisk}
\end{equation}
is minimized to find an estimate of the optimal predictor ${\hat p}^*({\bf x})$
over  a training set ${ D} = \{({\bf x}_1,y_1), \ldots, ({\bf x}_n,y_n)\}$.
Finally,   the  probabilities $\hat \eta({\bf x})$ are now estimated  from  ${\hat p}^*$ as
\begin{equation}
  {\hat \eta}({\bf x}) = [f^*_\phi]^{-1}({\hat p^*}({\bf x})).
  \label{eq:hateta}
\end{equation}

A predictor is denoted \emph{calibrated} \citep{DeGroot, Platt, Caruana, Raftery} if it is optimal, i.e. minimizes the risk of (\ref{eq:risk}), and
an optimal link function exists such that 
\begin{eqnarray}
{\hat \eta}({\bf x}) = { \eta}({\bf x}).
\end{eqnarray}

\subsubsection{Proper Losses and Probability Elicitation }
Risk minimization is closely related to  probability elicitation  ~\citep{Savage,DeGroot}. This connection has  been studied in ~\citep{Buja,HamedNunoLossDesign,Reid}.
In  probability elicitation, we assume that a forecaster makes a forecast ${\hat \eta}$ that an 
event $y=1$ will happen. The forecaster can be dishonest and not produce the true chance of event $y=1$.
The true chance of event $y=1$ when a forecast of  ${\hat \eta}$ is made is  denoted by $\eta$ and can be found as
\begin{eqnarray}
\eta&=&\frac{\mbox{Number of times the forecaster predicted}~ {\hat \eta} ~\mbox{and event}~ y=1 ~\mbox{happened} }{ \mbox{Total number of times that  the forecaster predicted} ~{\hat \eta}  } \\
&=& P(y=1|{\hat \eta}) 
\end{eqnarray}
If the forecaster is honest and always reports ${\hat \eta}=\eta$ then we say the forecaster is \emph{calibrated}.
 

The question of evaluating a forecaster can be answered by considering \emph{proper scoring functions}.  
A score function  assigns a score of $I_1({\hat \eta})$   to prediction ${\hat \eta}$ 
when event $y=1$ holds and a score of
$I_{-1}({\hat \eta})$ to prediction ${\hat \eta}$ when $y=-1$ holds. The scoring function is said to be proper if 
$I_1$ and $I_{-1}$ are such that the 
expected score for a given ${\hat \eta}$
\begin{equation}
  I(\eta,{\hat \eta}) = \eta I_{1}({\hat \eta}) + (1-\eta) I_{-1}({\hat \eta}),
  \label{eq:expreward}
\end{equation}
is maximal when ${\hat \eta} = \eta$, in other words
\begin{equation}
  I(\eta,{\hat \eta}) \leq I(\eta, \eta) = J(\eta), \,\,\, \forall \eta
  \label{eq:Savagebound}
\end{equation}
with equality if and only if ${\hat \eta} = \eta$. 
In other words, the score is maximized only if the forecaster is always honest.

The following theorem provides a means of generating proper scoring functions.
\begin{Thm}{~\citep{Savage}}
  Let $I(\eta,{\hat \eta})$ be as defined 
  in~(\ref{eq:expreward}) and $J(\eta) = I(\eta,\eta)$.
  Then~(\ref{eq:Savagebound}) holds if and only if $J(\eta)$ is convex and 
  \begin{equation}
    \label{eq:Is}
    I_1(\eta) = J(\eta) + (1-\eta) J^\prime(\eta) \quad \quad \quad
    I_{-1}(\eta) = J(\eta) -\eta J^\prime(\eta).
  \end{equation}
  \label{thm:savage}
\end{Thm}
Not surprisingly, proper losses are related to  probability elicitation and proper scoring functions by the following theorem.
\begin{Thm}{~\citep{HamedNunoLossDesign}}
\label{Thm:HamedNuno} Let $I_1(\cdot)$ and 
  $I_{-1}(\cdot)$ be as in (\ref{eq:Is}),
  for any continuously differentiable convex $J(\eta)$ such that
  $J(\eta) = J(1-\eta)$, and $f(\eta)$ any invertible function such that
  $f^{-1}(-v) = 1 -  f^{-1}(v)$.
  Then
  \begin{equation*}
    I_1(\eta) = -\phi(f(\eta)) \quad \quad \quad \quad \quad \quad 
    I_{-1}(\eta) = -\phi(-f(\eta)) \label{eq:I1I-1f}
  \end{equation*}
  if and only if 
  \begin{equation*}
    \phi(v) =  -J\left(f^{-1}(v)\right) - (1- f^{-1}(v)) J^\prime 
    \left(f^{-1}(v)\right).
    \label{eq:phieq}
  \end{equation*}
  \label{thm:risk}
\end{Thm}


This means that for any continuously differentiable $J(\eta) = -C_\phi^*(\eta)$ and invertible 
$ f^*_\phi(\eta)$,  the conditions of  Theorem \ref{Thm:HamedNuno} are satisfied and so the loss is 
\begin{equation}
  \phi(v) =  C_\phi^*\left([f_\phi^*]^{-1}(v)\right) + (1- [f_\phi^*]^{-1}(v)) 
  [C_\phi^*]^\prime\left([f_\phi^*]^{-1}(v)\right)
  \label{eq:phieq2}
\end{equation}
 and $I(\eta,{\hat \eta}) = -C_\phi(\eta,f)$.

The expected score of the forecaster over  $\hat \eta$ and $y$ is denoted by $S_{I_y}$ and can be written as
\begin{eqnarray}
S_{I_y}&&=\int_{\hat \eta} P_{\hat H}(\hat \eta) \sum_y P(y|\hat \eta)I_{y}(\hat \eta) d(\hat \eta)  \\
&&=\int_{\hat \eta} P_{\hat H}(\hat \eta) \Bigl( P(1|\hat \eta)I_1(\hat \eta) + P(-1|\hat \eta)I_{-1}(\hat \eta))\Bigr) d(\hat \eta). \nonumber \\
&&= \int_{\hat \eta} P_{\hat H}(\hat \eta)[\eta I_1(\hat \eta)+(1-\eta)I_{-1}(\hat \eta)] d(\hat \eta), \nonumber
\end{eqnarray}
where $P_{\hat H}(\hat \eta)$ is the distribution of $\hat \eta$.

We can dissect the expected score $S_{I_y}$ of a forecaster
into two parts of $S_{Calibration}$ and $S_{Refinement}$ that are measures of \emph{calibration} and \emph{refinement} (also known as \emph{sharpness}) \citep{DeGroot}.
This is accomplished by adding and subtracting 
$P_{\hat H}(\hat \eta)[P(1|\hat \eta)I_1(\eta) + P(-1|\hat \eta)I_{-1}(\eta)]$ as 
\begin{eqnarray}
&& S_{I_y}=\int_{\hat \eta} P_{\hat H}(\hat \eta) \sum_y P(y|\hat \eta)I_{y}(\hat \eta) d(\hat \eta)  \\
&&=\int_{\hat \eta} P_{\hat H}(\hat \eta) \Bigl( P(1|\hat \eta)I_1(\hat \eta) + P(-1|\hat \eta)I_{-1}(\hat \eta))\Bigr) d(\hat \eta) \nonumber \\ 
&&= \int_{\hat \eta} P_{\hat H}(\hat \eta) \Bigl[ P(1|\hat \eta)\Bigl\{I_1(\hat \eta)-I_1(\eta)\Bigr\} +  
P(-1|\hat \eta)\Bigl\{I_{-1}(\hat \eta)-I_{-1}(\eta) \Bigr\} \Bigr] d(\hat \eta)  \nonumber  \\ 
&& + \int_{\hat \eta} P_{\hat H}(\hat \eta) \Bigl[P(1|\hat \eta)I_1(\eta) + P(-1|\hat \eta)I_{-1}(\eta)\Bigr] d(\hat \eta) \nonumber \\
&&=S_{Calibration}+S_{Refinement}. \nonumber
\end{eqnarray}
Note that  $S_{Calibration}$ has a maximum equal to zero when the forecaster is calibrated ($\hat \eta=\eta$) and is negative otherwise since 
$I(\hat \eta,\eta) \le I(\eta,\eta)=J(\eta)$.

Also, for a calibrated forecaster, the $S_{Refinement}$ term can be simplified to 
\begin{eqnarray}
\label{eq:RefOrig}
S_{Refinement}\!\!\!\!\!\!&&=\int_{ \eta} P_{ H}( \eta) \Bigl[P(1| \eta)I_1(\eta) + P(-1| \eta)I_{-1}(\eta)\Bigr] d( \eta) \\
&&=\int_{ \eta} P_{ H}( \eta) \Bigl[\eta I_1(\eta) + (1-\eta) I_{-1}(\eta)\Bigr] d( \eta) \nonumber \\
&&=\int_{ \eta} P_{ H}( \eta) J(\eta) d( \eta). \nonumber  \\
\end{eqnarray}
Recall that $J(\eta)$ is  a convex function of   $ \eta$ over the $[0~1]$ interval with maximum values at $1$ and $0$. 
This means that if the distribution of the forecasters predictions $P_{ H}( \eta)$ is  more concentrated around $0$ and $1$, thus increasing 
the $P_{ H}( \eta)J(\eta)$ term, then $S_{Refinement}$ will increase and the forecaster is said to be more refined \citep{DeGroot}.

Finally, if the forecaster is calibrated then $S_{Refinement}$ can be written as
\begin{eqnarray}
\label{eq:RefOrigCali}
S_{Refinement}=\int_{ \eta} P_{ H}( \eta)J(\eta) d( \eta)  =E_{H}[J(\eta)]. 
\end{eqnarray} 
In summary, a calibrated and highly refined forecaster is ideal in the sense that it only makes honest definite forecasts of $\eta=0$ or  $\eta=1$ and is always
correct.

\section{Ensemble Learning Methods }
The idea of ensemble learning is to combine multiple classifier predictors to form an ensemble predictor with improved performance. 
The difference between  ensemble learning algorithms is in the assumptions they make about how the original multiple predictors 
vary. This in turn, dictates how they should be combined to complement each other and form an improved ensemble predictor. 
In this section we review two general ensemble algorithms namely, Bagging and Boosting. We also specifically consider the  AdaBoost and RealBoost 
algorithms which are used in the experiments section. 

\subsection{Bagging}
In Bagging, we  assume that the multiple predictors 
vary in terms of the training data sets 
${D} = \{({\bf x}_1,y_1), \ldots, ({\bf x}_n,y_n)\} \in \cal D$ from which they were trained. We also assume that the data sets are sampled from an 
underlying distribution $P_{{\bf D}}(D)$.  To emphasis this point we can write a  predictor as $p({\bf x}; D)$. The ensemble predictor
is taken to be the average of the different predictors over the  training data set distribution written as
\begin{equation}
p_e({\bf x}; {\cal D})=E_{{\bf D}}[p({\bf x}; D)].
\end{equation}
It can be shown 
in \citep{Bagging},  that  the  error of the ensemble predictor written as
\begin{equation}
E_e=E_{Y,{\bf X}}[(y-p_e({\bf x}; {\cal D}))^2],
\end{equation}
is less than the average error of the multiple predictors  written as
\begin{equation}
E=E_{{\bf D}}[ E_{Y,{\bf X}}[(y-p({\bf x}; D))^2] ],
\end{equation}
where $y$ is the true numerical prediction.
In other words, the ensemble predictor has on average, less error than the multiple predictors and is thus a more reliable predictor.

\subsection{Boosting}
In Boosting we assume that a set of multiple classifier predictors exists in the form of ${p({\bf x}^{(1)}), ..., p({\bf x}^{(N)})}$ that vary
in terms of the feature vectors that they use, namely ${\bf x}^{(i)} \in {\cal X}^{(i)}$.  The feature
vectors ${\bf x}^{(i)}$ are sampled from their associated probability distributions $P_{{\bf X}^{(i)}}({\bf x}^{(i)})$.  We also assume that a single training set
${D} = \{({\bf x}_1,y_1), \ldots, ({\bf x}_n,y_n)\} $ exists and ${\bf x}_i^{(j)}$ is the $j$th feature of the $i$th training data point. 
The Boosting ensemble predictor is an additive model in the form of 
\begin{equation}
p_e(\{ {\bf x}^{(1)}, ..., {\bf x}^{(N)} \})=\sum_{i=1}^M \alpha^{(i)} \cdot p({\bf x}^{(i)}).
\end{equation}
The $p({\bf x}^{(i)})$ and $\alpha^{(i)}$  are sequentially added to the model so as  to minimize the risk of the updated ensemble. In other words, Boosting 
adds the best predictors and features to the model such that they compliment each other and improve the resulting ensemble predictors performance.

Boosting algorithms generally differ based on the specific loss function used to compute the risk or the type of predictors $p({\bf x}^{(i)})$ used or the method used to
find the best update at each iteration. 
Here we  consider two commonly used boosting algorithms namely the AdaBoost and RealBoost algorithms. Both algorithms are based on the exponential loss and
find the update through a gradient descent method. The main difference between the two algorithms is in how the predictors are used.
The details of the AdaBoost and RealBoost algorithms are presented in Algorithm-\ref{algo:AdaBoost} and Algorithm-\ref{algo:RealBoost} respectively.

\begin{algorithm}[tb]
\caption{AdaBoost} \label{algo:AdaBoost}
{\small
\begin{algorithmic}
\STATE {\bf Input:} Training set ${D} = \{({\bf x}_1,y_1), \ldots, ({\bf x}_n,y_n)\}$, 
$y_i \in \{1,-1\}$ is the class
label of example ${\bf x}_i$,  a set of classifier predictors ${p({\bf x}^{(1)}), ..., p({\bf x}^{(N)})}$ where $p({\bf x}^{(j)}) \in \{1,-1\}$
and  $M$ number of predictors in
the final decision rule. 
\STATE {\bf Initialization:} Set $w_i=\frac{1}{n}$.
\FOR{$m=1:M$} 
\STATE Find best update as  
\begin{eqnarray*}
p({\bf x}^{(m)})=\arg\min_{j} \sum_{i=1}^{n} w_i \cdot I[y_i \ne p({\bf x}_i^{(j)})],
\end{eqnarray*}
where $I$ is the identity function.
\STATE Compute
\begin{eqnarray*}
\epsilon=\frac{\sum_{i=1}^{n} w_i \cdot I[y_i \ne p({\bf x}_i^{(m)})]}{\sum_{i=1}^{n} w_i}  
\end{eqnarray*}
and
\begin{eqnarray*}
\alpha^{(m)}=\frac{1}{2}\log \left(\frac{1-\epsilon}{\epsilon} \right).
\end{eqnarray*}
\STATE Change $w_i$ as
\begin{eqnarray*}
w_i=w_i \cdot e^{ \alpha^{(m)}  I[y_i \ne p({\bf x}_i^{(m)})] }
\end{eqnarray*}
\ENDFOR 
\STATE {\bf Output:} ensemble predictor  $\sum_{m=1}^M \alpha^{(m)} \cdot p({\bf x}^{(m)})$.
\end{algorithmic}
}
\normalsize
\end{algorithm}

\begin{algorithm}[tb]
\caption{RealBoost} \label{algo:RealBoost}
{\small
\begin{algorithmic}
\STATE {\bf Input:} Training set ${D} = \{({\bf x}_1,y_1), \ldots, ({\bf x}_n,y_n)\}$, 
$y_i \in \{1,-1\}$ is the class
label of example ${\bf x}_i$,  a set of classifier predictors ${p({\bf x}^{(1)}), ..., p({\bf x}^{(N)})}$ where $p({\bf x}^{(j)}) \in \mathbb{R}$
and  $M$ number of predictors in
the final decision rule. 
\STATE {\bf Initialization:} Set $w_i=\frac{1}{n}$.
\FOR{$m=1:M$} 
\STATE Find best update as  
\begin{eqnarray*}
p({\bf x}^{(m)})=\arg\min_{j} \sum_{i=1}^{n} w_i \cdot e^{ -y_i  p({\bf x}_i^{(j)}) }.
\end{eqnarray*}
\STATE Change $w_i$ as
\begin{eqnarray*}
w_i=w_i \cdot e^{ -y_i  p({\bf x}_i^{(m)}) }
\end{eqnarray*}
and normalize $w_i$ such that 
\begin{eqnarray*}
\sum_{i=1}^{n} w_i=1.
\end{eqnarray*}
\ENDFOR 
\STATE {\bf Output:} ensemble predictor  $\sum_{m=1}^M  p({\bf x}^{(m)})$.
\end{algorithmic}
}
\normalsize
\end{algorithm}

%


\section{Combining Forecasters and Ensemble Learning }
In this section we make the connection between ensemble learning and combining forecasters. We start from the basic axioms of probability and construct the combining forecaster problem such that it can be readily viewed as an ensemble learning problem. 

The logical axioms of probability, namely that of \emph{consistency} \citep{book:Jaynes} , require that if two forecasters have the same information then it is only logical that they should provide the same  forecast. Different forecasters differ in their output because they  consider different information. Formally, we denote information to be ${{\bf I}}$ and assume that it is comprised of a vector of features ${\bf x \in \cal X}$ and an independently drawn training data set $D=\{ (y_1,{\bf x}_1), ..., (y_n,{\bf x}_n)  \} \in \cal D$. Specifically, the feature vectors and training data set are drawn from the independent probability distributions of $P_{\bf X}({\bf x})$ and $P_{\bf D}(D)$ respectively. A probability forecaster provides the probability of event $y$ given the information ${\bf I}$ and is  written as 
\begin{eqnarray}
\eta({\bf I})=\eta({\bf x}; D)=P(y|{\bf I}). 
\end{eqnarray}
A point forecaster  $p({\bf I})$ can now be treated as a classifier predictor and written as 
\begin{eqnarray}
p({\bf I})=f^*(\eta({\bf I}))=f^*(P(y|{\bf I})) 
\end{eqnarray}
for an appropriate link function such as those of (\ref{eq:fexamples}). 
Conversely, a point forecast can be transformed into a probability forecast by writing 
\begin{eqnarray}
[f^*]^{-1}(p({\bf I}))=\eta({\bf I}).
\end{eqnarray}

We say two forecasters have a type-I difference if they only differ in their training data set and a type-II difference if they only differ in their
feature vectors.  As an example of a type-II difference, consider the problem of forecasting the probability of rain.  The vector of features for one rain forecaster could be air pressure readings while another forecaster could use another feature such as air temperature readings. Obviously, these two forecasters will provide different probability of rain forecasts. As a type-I example, both forecasters could use the same feature vector in the form of air pressure readings but have access to different training data sets. For example the first forecaster could have access to the past $7$ day air pressure readings  while the second forecaster could have access to the past $365$ day air pressure readings. It is again obvious that these two forecasters will provide different probability of rain forecasts based on the different information they have. 

Under this model, the issue of having multiple forecasters that are all calibrated is now easily resolved. Different forecasters can provide different forecasts while still being calibrated because each forecaster is based on a different feature or training set with a different associated probability distribution. In other words, they can provide different forecasts while each still being calibrated with respect to their own different probability distributions. 

  
Under this model, combining forecasters can now be readily posed as an ensemble learning problem. Specifically, in ensemble learning a combination of multiple 
predictors are used to produce an ensemble predictor that has higher accuracy than any of the individual predictors. This is exactly the same goal of 
combining forecasters when the term predictor is  replaced  with the term forecaster. In other words, ensemble learning can be used to find an ensemble or combination of forecasters with improved performance.  With this connection  in place, we next show how different well established ensemble learning methods such as Bagging and Boosting can be converted for combining forecasters.

\section{Combining Forecasters Using Bagging}

In terms of combining forecasters, Bagging can be thought of as a method that deals  with type-I differences among  forecasters. Specifically, let $\eta({\bf x}; D_1)$ through $\eta({\bf x}; D_N)$ be $N$ probability forecasters that have  type-I differences based on the  different training data sets of $D_1$ through $D_N$. The training data sets  $D_i=\{ (y_1,{\bf x}_1), ..., (y_n,{\bf x}_n)  \}$ are each comprised of  independent observations from the same underlying distribution.  We want to somehow combine these probability forecasters to form an ensemble forecaster $\eta_e({\bf x}; {\cal D})$ such that the  mean squared error of the ensemble forecaster 
is less than the average  mean squared error of the $N$ individual forecasters.
Bagging suggests that this can be accomplished by simply averaging 
the $N$ probability forecasters as
\begin{equation}
\eta_e({\bf x}; {\cal D})=\frac{1}{N} \sum_i^N \eta({\bf x}; D_i).
\label{eq:AvgPropsForm}
\end{equation}  

Similarly, point forecasters $p({\bf x}; D_1)$ through $p({\bf x}; D_N)$ that have type-I differences can also be combined to 
form an ensemble forecaster $\eta_e({\bf x}, {\cal D})$ as
\begin{equation}
p_e({\bf x}; {\cal D})=\frac{1}{N} \sum_i^N p({\bf x}; D_i).
\end{equation}  

It is interesting to note that equation (\ref{eq:AvgPropsForm}) is the well known method of combining probability forecasters 
by averaging their forecasts  \citep{Timmermann1, Clements}.
While there has been much debate regarding this  method \citep{Ranjan}, the above Bagging interpretation provides an elegant explanation of  why this simple method works so well in many cases.  

Finally, note that an ensemble classifier can be found from multiple probability forecasters using equation (\ref{eq:ResultClassifier})  as
\begin{equation}
h_e({\bf x}; {\cal D})=sign(f^*(\eta_e({\bf x}; \cal D))).
\end{equation}  
Similarly an ensemble classifier can also be found from multiple point forecasters as
\begin{equation}
h_e({\bf x}; {\cal D})=sign(p_e({\bf x}; \cal D)).
\end{equation}

\section{Combining Forecasters Using Boosting}
Boosting can be thought of as a method that deals  with type-II differences among  forecasters. Specifically, let $\eta({\bf x}^{(1)})$ 
through $\eta({\bf x}^{(N)})$
be $N$ probability forecasters that have  type-II differences based on the  different feature vectors of ${\bf x}^{(1)}$ through ${\bf x}^{(N)}$. We want to somehow combine these individual probability forecasters to form an ensemble forecaster $\eta_e(\{ {\bf x}^{(1)}, ..., {\bf x}^{(N)} \})$ such that the  risk associated with the ensemble forecaster  
is less than the risk associated with any of the $N$ individual forecasters. Boosting suggests that the ensemble probability forecaster has the form of
\begin{equation}
\eta_e(\{ {\bf x}^{(1)}, ..., {\bf x}^{(N)} \})= [f^*]^{-1}\left( \sum_{i=1}^M \alpha^{(i)} \cdot p({\bf x}^{(i)}) \right),
\label{eq:BoostForm}
\end{equation}  
where
\begin{equation}
p({\bf x}^{(i)})=f^*\left(\eta({\bf x}^{(i)})\right)
\end{equation}
is an associated point forecaster and the weights $\alpha^{(i)}$ are found from the specific Boosting algorithm that is being used. 
Also, the ensemble classifier can be found from multiple point forecasters and written as
\begin{equation}
h_e(\{ {\bf x}^{(1)}, ..., {\bf x}^{(N)} \})=  sign \left[ \sum_{i=1}^M \alpha^{(i)} \cdot p({\bf x}^{(i)}) \right].
\label{eq:BoostFormClassifier}
\end{equation}  

Note that, the ensemble forecaster in equation (\ref{eq:BoostForm}) is found by \emph{calibrating} the ensemble classifier  using an appropriate inverse link function 
 $[f^*]^{-1}$ as suggested in \citep{HamedNunoJMLRRegularize, Platt, Caruana, Raftery}. 

While Bagging combines the forecasters using a simple linear  combination of the individual forecasters,  Boosting does so using  a highly nonlinear combination of the  forecasters. This nonlinear combining of the forecasters allows the resulting ensemble forecaster to span regions of the space of forecaster functions that a simple linear combination of forecasters cannot possibly span. This in turn allows for a more flexible ensemble forecaster with potentially improved performance. This is further confirmed in the  experimental results. The nonlinear nature of the Boosting ensemble forecaster can thus allow it to escape many of the theoretical limitations  \citep{Ranjan} seen in the traditional methods of linear combination of forecasters.

\subsection{The Boosted  Forecaster is Both Calibrated and Highly Refined}
We have shown that combining forecasters can be posed as an ensemble learning problem and solved using the Boosting algorithm. The Boosting algorithm has been extensively studied and  theoretically justified  \citep{freund, friedman} and so there is no need to reiterate these theoretical arguments here for the special case of when Boosting is used for combining forecasters. Nevertheless, we can directly show that  Boosting  leads to a valid forecaster .
This means a forecaster that is both calibrated and has high refinement (sharpness).

We have seen that a predictor is denoted \emph{calibrated} 
if it is optimal, i.e. minimizes the risk of (\ref{eq:risk}),
and an optimal link  exists associated with a proper loss function such that
\begin{eqnarray}
\eta({\bf x})=(f^*)^{-1}(p^*({\bf x})).
\end{eqnarray} 
In a classification algorithm a proper loss function is usually fixed beforehand. In the case of AdaBoost and RealBoost this is the exponential loss of 
\begin{eqnarray}
L(p({\bf x}), y)=e^{-yp({\bf x})}
\end{eqnarray} 
which is verified to be a proper loss function \citep{zhang}  with an associated optimal link of 
\begin{eqnarray}
f^*(\eta)=\frac{1}{2}\log(\frac{\eta}{1-\eta}). 
\end{eqnarray}
This directly demonstrates that the predictor found from Boosting is theoretically calibrated. 
Other studies have confirmed that the Boosting algorithm can  produce calibrated predictors in practice as well \citep{HamedNunoJMLRRegularize}.

The risk of the ensemble predictor found from Boosting is less than the risk of the individual predictors by design. In what follows, we show that the refinement of the ensemble forecaster of Boosting is greater than the refinement of the individual forecasters. This means that the ensemble forecaster of Boosting is highly refined. To show this, we first establish a direct connection between 
minimum risk and refinement through the following  two lemmas.
\begin{lemma}
\label{thm:MinRiskTermForm}
Let $\phi$ be a proper loss function  and $p^*$ the optimal
predictor. The minimum risk $R_{\phi}(p^*)$ can be written in the form of
\begin{eqnarray}
\label{eq:MinRiskClean}
R_{\phi}({p^*})=\int_{\bf x} P_{{\bf X}}({\bf x})  C_\phi^*(\eta({\bf x} ) ) d{\bf x}.
\end{eqnarray}
\end{lemma}
\begin{proof}
From Theorem \ref{Thm:HamedNuno} and using equations (\ref{equ:MinRiskInit}) and $\eta({\bf x})=P_{Y|{\bf X}}(1|{\bf x})={[f_\phi^*]}^{-1}(p^*({\bf x})) $, we write
\begin{eqnarray}
&&R_{\phi}({p^*}) = \int_{\bf x} P_{{\bf X}}({\bf x}) \left[ P_{{\bf Y}|{\bf X}}(1|{\bf x})\phi({ p^*}({\bf x})) + P_{{\bf Y}|{\bf X}}(-1|{\bf x})\phi(-{ p^*}({\bf x}))  
\right] d{\bf x} \\
&=& \int_{\bf x} P_{{\bf X}}({\bf x}) [ \eta({\bf x}) C_\phi^*(\eta({\bf x}) ) + \eta({\bf x})(1-\eta({\bf x})) [C_\phi^*]^\prime\left(\eta({\bf x})\right) \\    
&+& (1-\eta({\bf x})) C_\phi^*((1-\eta({\bf x})) ) + (1-\eta({\bf x}))(\eta({\bf x})) [C_\phi^*]^\prime((1-\eta({\bf x}))) ] d{\bf x} \\
&=& \int_{\bf x} P_{{\bf X}}({\bf x}) [ \eta({\bf x}) C_\phi^*(\eta({\bf x}) ) + \eta({\bf x})(1-\eta({\bf x})) [C_\phi^*]^\prime\left(\eta({\bf x})\right) \\     
&+&  C_\phi^*(\eta({\bf x}) ) - \eta({\bf x}) C_\phi^*(\eta({\bf x}) ) - \eta({\bf x})(1-\eta({\bf x})) [C_\phi^*]^\prime\left(\eta({\bf x})\right)  ] d{\bf x} \\
&=& \int_{\bf x} P_{{\bf X}}({\bf x})  C_\phi^*(\eta({\bf x}) ) d{\bf x}.
\end{eqnarray}
\end{proof}

\begin{lemma}
\label{thm:MinRiskequalRefinement}
Let $\phi$ be a proper loss function  and $p^*$ the optimal
predictor with associated  minimum risk $R_{\phi}(p^*)$. Then, the minimum risk  is equal to the negative refinement of the associated calibrated forecaster and we write
\begin{eqnarray}
R_{\phi}(p^*) = - S_{Refinement}.
\end{eqnarray}
\end{lemma}
\begin{proof}
Noting that $C_\phi^*(\eta)=-J(\eta)$ and Lemma \ref{thm:MinRiskTermForm} we have
\begin{eqnarray}
&R_{\phi}({p^*})&=\int_{\bf x} P_{{\bf X}}({\bf x})  C_\phi^*(\eta({\bf x} ) ) d{\bf x} \\
&=& - \int_{\bf x} P_{{\bf X}}({\bf x})  J(\eta({\bf x} ) ) d{\bf x} \\
&=& -E_{X}[J(\eta({\bf x} ))].
\end{eqnarray}
Denoting $J(\eta({\bf x} ))=q({\bf x})$, $\eta({\bf x})=z$ and using the  expected value of a function  of a random variable theorem  we continue to write
\begin{eqnarray}
-E_{X}[J(\eta({\bf x} ))]& = &-E_{X}[q({\bf x})] = -E_{Z}[q({\bf x})]  
= -E_{Z}[J(z)] = - S_{Refinement}.
\end{eqnarray}
\end{proof}

Finally, we can now readily prove that the Boosted ensemble forecaster has improved refinement in the following theorem.
\begin{Thm}
\label{thm:EnsembleRefinement}
Let $\eta_e(\{ {\bf x}^{(1)}, ..., {\bf x}^{(N)} \})$ be the ensemble forecaster found from Boosting and $\eta({\bf x}^{(i)})$ be any individual forecaster. Then, the refinement of the ensemble forecaster is greater than or equal to the refinement of any of the individual forecasters and we can write
\begin{eqnarray}
S_{Refinement}(\eta_e(\{ {\bf x}^{(1)}, ..., {\bf x}^{(N)} \})) \ge S_{Refinement}(\eta({\bf x}^{(i)}))  ~~\forall i .
\end{eqnarray}
\end{Thm}
\begin{proof}
We know from the Boosting algorithm that the minimum risk associated with the ensemble predictor is less than or equal to the minimum risk of any  individual predictor
so we can write
\begin{eqnarray}
R(p^*_e(\{ {\bf x}^{(1)}, ..., {\bf x}^{(N)} \})) \le R(p^*({\bf x}^{(i)}))  ~~\forall i.
\end{eqnarray}
Using Lemma-\ref{thm:MinRiskequalRefinement} we have
\begin{eqnarray}
-S_{Refinement}(\eta_e(\{ {\bf x}^{(1)}, ..., {\bf x}^{(N)} \})) \le -S_{Refinement}(\eta({\bf x}^{(i)}))  ~~\forall i, 
\end{eqnarray}
from which  the proof follows.
\end{proof}

%
%

\section{Experimental Results}
\label{sec:Exp}
In this section we perform experiments on the Good Judgment Project data set \citep{GoodJudgmentProj}
and compare the different ensemble forecasters learned from Bagging, AdaBoost and RealBoost.
We first discuss the Good Judgment Project  in general and the specifics of the data set that we used.

\subsection{The Good Judgment Project Data Set}
The Good Judgment Project data set consists of a wide variety of mostly economic and political questions such as ``Will OPEC agree to cut its oil output at or before its 27 November 2014 meeting?'' or ``Who will win the April 2014 presidential elections in Afghanistan?'' and a set of predictions provided by different people who were asked to provide a probability forecast for the event. Once the outcome of the question had materialized, the different forecasters were given a score in accordance with their forecasting ability and the winners of the contest were declared. The contest spanned a number of years and forecasters were grouped as untrained, trained, super forecasters (winners of previous contests) and team forecasters consisting of multiple  forecasters. 
In our experiments we only considered the untrained forecaster data set and extracted  $88$ unique binary questions and a total of $338$ forecasters who had answered these questions along with their respective outcomes.  

\subsection{Comparing Different Forecasters}
In accordance with the model presented in this paper, different forecasters  base their forecasts on the different data they can access  or the different features
they consider. For example, different data might mean  having the previous years OPEC oil output or not. An example for different features could be a case where one forecaster decides to use economic factors while another forecaster might base his or her forecast on social factors alone. Intuitively, a good combination of forecasters could do better than any individual forecaster because a combination of forecasters is taking advantage of all the available information in terms of all the data and features considered by the multiple forecasters.

In our experiments we first established a performance baseline by finding the number of errors for each individual forecaster and then finding the average error over all forecasters and the best individual forecaster with smallest error. The error for each forecaster was computed by considering the probability forecast provided by the forecaster for each question. If the probability for an event provided by the forecaster was above $0.5$ and the event materialized then there was no error and if the event did not materialize then there was an error. This is basically implementing the Bayes decision rule. An error was also counted if a forecaster failed to provide a probability. In this setting the highest possible error for any forecaster is $88$ meaning that the forecaster failed to correctly predict all $88$ questions. Under this setting, the average error over all forecasters was
$76.11$ and the best individual forecaster had an error of $30$. One reason for such higher than chance errors is that many forecasters simply failed to provide a probability forecast for most of the questions.

Next, we considered combining the multiple forecasters to form an ensemble forecaster using Bagging.  Under this setting, the ensemble forecaster's prediction for each question is simply  the average of the probabilities reported by all forecasters for that question. If an individual  forecaster did not provide a probability for a specific question, a forecast of $0.5$ was assumed. If the  ensemble forecaster's probability forecast was above $0.5$ and the event materialized then there was no error and if the event did not materialize then there was an error. The Bagging ensemble forecaster had an error of $9$ meaning that the Bagging ensemble forecaster made a wrong prediction on only $9$ out of the $88$ questions. This is a significant improvement over any individual forecaster which   had at best an error of $30$. 

Finally, we considered  combining the multiple forecasters to form an ensemble forecaster using Boosting. We trained the Boosting ensemble forecaster using a leave one out approach. Specifically, we trained a Boosting ensemble forecaster for each question while using the remaining $87$ questions as training data.  In other words, we assume that we have access to the probability forecast provided by each forecaster on the $87$ previous questions. We will use the Boosting algorithm to find the 
best combination of forecasters based on their previous performance history.  We made the ensemble classifier using both the AdaBoost and RealBoost algorithms. 
In the AdaBoost algorithm we converted the probabilities provided by the forecasters into the binary classifier predictors required by the AdaBoost algorithm using the following link function
\begin{equation}
p({\bf x})=sign(2\eta({\bf x})-1).
\end{equation}
Also, if a forecaster did not provide a probability for a certain question, a random probability was assigned to it in the AdaBoost algorithm.
In the RealBoost Algorithm we converted the probabilities provided by the forecasters into the real classifier predictors required by the RealBoost algorthim using the following link function
\begin{equation}
p({\bf x})=\frac{1}{2}\log\left(\frac{\eta({\bf x})}{1-\eta({\bf x})}\right).
\end{equation}
Also, if a forecaster did not provide a probability for a certain question, a  probability of $0.5$ was assigned to it in the RealBoost algorithm.
The AdaBoost ensemble forecaster was trained for $800$ iteration and had an error of $7$. 
The RealBoost ensemble forecaster was trained for only $70$ iterations and had an error of $6$.   

We summarize all of the above results in Table \ref{tab:PredError}. In general, combining the forecasters using either Bagging or Boosting will significantly improve the results. Among the different ensemble methods, RealBoost has the best performance followed by AdaBoost. Bagging has slightly more errors when compared to the Boosting methods, but has the benefit of being very simple to train and does not require having previously acquired a performance history  for the forecasters.  
We can also compare the ensemble methods in terms of the number of forecasters needed to form a prediction. Bagging requires that all $338$ forecasters provide
a forecast since its forecast is the average of all the forecasts. Boosting on the other hand has the benefit of not requiring that all forecasters participate
and actually discards the under-performing forecasters. AdaBoost only requires an average of about $191$ unique forecasters and discards the other $147$ forecasters. Interestingly, RealBoost only requires an average of about $26$ unique forecasters and discards the other $312$ forecasters. In other words, among the $338$ individual forecasters, an ensemble of only $26$ of the best forecasters that compliment each other are found by the RealBoost algorithm and have the best predictive performance
with only $6$ prediction errors.

\begin{table*}[t]
\caption{Prediction error for different methods.  }
\centering
\begin{tabular}{|c|c|c|c|  }
    \hline 
    Method & Prediction Error & Avg. \# of Unique Forecasters Used   \\
    \hline
		\hline
		Best Individual Forecaster & 30 & 1 \\
    \hline
		Bagging & 9 & 338 \\
    \hline
		AdaBoost & 7 & 191 \\
    \hline
		RealBoost & {\bf 6} & {\bf 26} \\
    \hline
\end{tabular}
\label{tab:PredError}
\end{table*}


\section{Conclusion}
\label{sec:conc}
In this work, we considered the problem of combining forecasters and showed that ensemble learning methods can be readily adapted for this purpose. This opens up a host of methods for theoretically analyzing and practically designing strategies for combining forecasters. For example, we considered the Bagging and Boosting ensemble learning algorithms and adapted them for combining forecasters. The Bagging algorithm when adapted for combining forecasters, reduced to the simple method of averaging forecasts and its analysis provided theoretical insight into the averaging methods success. Bagging does not require training data in the form of a forecaster’s performance history but leads to a combined forecaster with weaker performance guarantees and requires the combination of a larger number of individual forecasters. Boosting takes advantage of the training data in the form of the individual forecaster’s performance history and provides strong theoretical guarantees in terms of calibration and high refinement. An added benefit is that it discards the weak individual forecasters and builds a nonlinear combination of forecasters using a small number of the best individual forecasters available. In any case, this paper serves as a gateway for inviting a multitude of well established and powerful methods of ensemble learning into the world of combining forecasters.

\bibliographystyle{Chicago}
\bibliography{IEEEexample}

\begin{thebibliography}{}

\bibitem[\protect\citeauthoryear{Baars and Mass}{Baars and Mass}{2005}]{Baars}
Baars, J.~A. and C.~F. Mass (2005).
\newblock Performance of national weather service forecasts compared to
  operational, consensus, and weighted model output statistics.
\newblock {\em Weather and Forecasting\/}~{\em 20}, 1034–1047.

\bibitem[\protect\citeauthoryear{Bartlett, Jordan, and McAuliffe}{Bartlett
  et~al.}{2006}]{JordanBartlett}
Bartlett, P., M.~Jordan, and J.~D. McAuliffe (2006).
\newblock Convexity, classification, and risk bounds.
\newblock {\em JASA\/}.

\bibitem[\protect\citeauthoryear{Breiman}{Breiman}{1996}]{Bagging}
Breiman, L. (1996).
\newblock Bagging predictors.
\newblock {\em Machine Learning\/}~{\em 24\/}(2), 123--140.

\bibitem[\protect\citeauthoryear{Brocker}{Brocker}{2009}]{Brocker2009}
Brocker, J. (2009).
\newblock Reliability, sufficiency, and the decomposition of proper scores.
\newblock {\em Quarterly Journal of the Royal Meteorological Society\/}~{\em
  135\/}(643), 1512--1519.

\bibitem[\protect\citeauthoryear{Buja, Stuetzle, and Shen}{Buja
  et~al.}{2005}]{Buja}
Buja, A., W.~Stuetzle, and Y.~Shen (2005).
\newblock Loss functions for binary class probability estimation and
  classification: Structure and applications.
\newblock {\em (Technical Report) University of Pennsylvania\/}.

\bibitem[\protect\citeauthoryear{Clements and Harvey}{Clements and
  Harvey}{2011}]{Clements}
Clements, M.~P. and D.~I. Harvey (2011).
\newblock Combining probability forecasts.
\newblock {\em International Journal of Forecasting\/}~{\em 27\/}(2), 208--223.

\bibitem[\protect\citeauthoryear{Croushore}{Croushore}{1993}]{Croushore}
Croushore, D.~D. (1993).
\newblock Introducing: The survey of professional forecasters.
\newblock {\em Federal Reserve Bank of Philadelphia Business Review\/}~{\em
  November/December}, 3--13.

\bibitem[\protect\citeauthoryear{DeGroot and Fienberg}{DeGroot and
  Fienberg}{1983}]{DeGroot}
DeGroot, M.~H. and S.~E. Fienberg (1983).
\newblock The comparison and evaluation of forecasters.
\newblock {\em The Statistician\/}~{\em 32}, 14--22.

\bibitem[\protect\citeauthoryear{Devroye, Györfi, and Lugosi}{Devroye
  et~al.}{1997}]{book:ProbPatRec}
Devroye, L., L.~Györfi, and G.~Lugosi (1997).
\newblock {\em A Probabilistic Theory of Pattern Recognition}.
\newblock New York: Springer.

\bibitem[\protect\citeauthoryear{Freund and Schapire}{Freund and
  Schapire}{1997}]{freund}
Freund, Y. and R.~Schapire (1997).
\newblock A decision-theoretic generalization of on-line learning and an
  application to boosting.
\newblock {\em Journal of Computer and System Sciences\/}~{\em 55}, 119–139.

\bibitem[\protect\citeauthoryear{Friedman, Hastie, and Tibshirani}{Friedman
  et~al.}{2000}]{friedman}
Friedman, J., T.~Hastie, and R.~Tibshirani (2000).
\newblock Additive logistic regression: A statistical view of boosting.
\newblock {\em Annals of Statistics\/}~{\em 28}, 337--407.

\bibitem[\protect\citeauthoryear{Gneiting, Balabdaoui, and Raftery}{Gneiting
  et~al.}{2007}]{Tilmann2007}
Gneiting, T., F.~Balabdaoui, and A.~E. Raftery (2007).
\newblock Probabilistic forecasts, calibration and sharpness.
\newblock {\em Journal of the Royal Statistical Society Series B\/}, 243--268.

\bibitem[\protect\citeauthoryear{Gneiting and Raftery}{Gneiting and
  Raftery}{2007}]{Raftery}
Gneiting, T. and A.~Raftery (2007).
\newblock Strictly proper scoring rules, prediction, and estimation.
\newblock {\em Journal of the American Statistical Association\/}~{\em 102},
  359–--378.

\bibitem[\protect\citeauthoryear{Graham}{Graham}{1996}]{Graham}
Graham, J.~R. (1996).
\newblock Is a group of economists better than one? than none?
\newblock {\em Journal of Business\/}~{\em 69}, 193–232.

\bibitem[\protect\citeauthoryear{Jaynes and Bretthorst}{Jaynes and
  Bretthorst}{2003}]{book:Jaynes}
Jaynes, E.~T. and G.~L. Bretthorst (2003).
\newblock {\em Probability Theory: The Logic of Science}.
\newblock Cambridge: Cambridge University Press.

\bibitem[\protect\citeauthoryear{Masnadi-Shirazi and
  Vasconcelos}{Masnadi-Shirazi and Vasconcelos}{2008}]{HamedNunoLossDesign}
Masnadi-Shirazi, H. and N.~Vasconcelos (2008).
\newblock On the design of loss functions for classification: theory,
  robustness to outliers, and savageboost.
\newblock In {\em Advances in Neural Information Processing Systems}, pp.\
  1049--1056. MIT Press.

\bibitem[\protect\citeauthoryear{Masnadi-Shirazi and
  Vasconcelos}{Masnadi-Shirazi and Vasconcelos}{2015}]{HamedNunoJMLRRegularize}
Masnadi-Shirazi, H. and N.~Vasconcelos (2015).
\newblock A view of margin losses as regularizers of probability estimates.
\newblock {\em The Journal of Machine Learning Research\/}~{\em 16},
  2751--2795.

\bibitem[\protect\citeauthoryear{Niculescu-Mizil and Caruana}{Niculescu-Mizil
  and Caruana}{2005}]{Caruana}
Niculescu-Mizil, A. and R.~Caruana (2005).
\newblock Obtaining calibrated probabilities from boosting.
\newblock In {\em Uncertainty in Artificial Intelligence}.

\bibitem[\protect\citeauthoryear{Platt}{Platt}{2000}]{Platt}
Platt, J. (2000).
\newblock Probabilistic outputs for support vector machines and comparison to
  regularized likelihood methods.
\newblock In {\em Adv. in Large Margin Classifiers}, pp.\  61--74.

\bibitem[\protect\citeauthoryear{Ranjan and Gneiting}{Ranjan and
  Gneiting}{2010}]{Ranjan}
Ranjan, R. and T.~Gneiting (2010).
\newblock Combining probability forecasts.
\newblock {\em Journal of the Royal Statistical Society B\/}~{\em 72}, 71--91.

\bibitem[\protect\citeauthoryear{Reid and Williamson}{Reid and
  Williamson}{2010}]{Reid}
Reid, M. and R.~Williamson (2010).
\newblock Composite binary losses.
\newblock {\em The Journal of Machine Learning Research\/}~{\em 11},
  2387--2422.

\bibitem[\protect\citeauthoryear{Satopää, Pemantle, and Ungar}{Satopää
  et~al.}{2016}]{Satopaa1}
Satopää, V.~A., R.~Pemantle, and L.~H. Ungar (2016).
\newblock Modeling probability forecasts via information diversity.
\newblock {\em Journal of the American Statistical Association\/}~{\em
  111\/}(516), 1623--1633.

\bibitem[\protect\citeauthoryear{Savage}{Savage}{1971}]{Savage}
Savage, L.~J. (1971).
\newblock The elicitation of personal probabilities and expectations.
\newblock {\em Journal of The American Statistical Association\/}~{\em 66},
  783--801.

\bibitem[\protect\citeauthoryear{Timmermann}{Timmermann}{2006}]{Timmermann1}
Timmermann, A. (2006).
\newblock Forecast combinations.
\newblock {\em Handbook of economic forecasting\/}~{\em 1}, 135--196.

\bibitem[\protect\citeauthoryear{Ungar, Mellers, Satopää, Tetlock, and
  Baron}{Ungar et~al.}{2012}]{GoodJudgmentProj}
Ungar, L.~H., B.~A. Mellers, V.~Satopää, P.~Tetlock, and J.~Baron (2012).
\newblock The good judgment project: A large scale test of different methods of
  combining expert predictions.
\newblock In {\em Association for the Advancement of Artificial Intelligence
  2012 Fall Symposium Series}. Univ. Pennsylvania, Philadelphia, PA.

\bibitem[\protect\citeauthoryear{Winkler and Poses}{Winkler and
  Poses}{1993}]{Winkler}
Winkler, R.~L. and R.~M. Poses (1993).
\newblock Evaluating and combining physicians' probabilities of survival in an
  intensive care unit.
\newblock {\em Management Science\/}~{\em 39}, 1526–1543.

\bibitem[\protect\citeauthoryear{Zhang}{Zhang}{2004}]{zhang}
Zhang, T. (2004).
\newblock Statistical behavior and consistency of classification methods based
  on convex risk minimization.
\newblock {\em Annals of Statistics\/}~{\em 32}, 56--85.

\end{thebibliography}
\end{document}